\documentclass[12pt]{article}

\usepackage{setspace}
\usepackage{fullpage}
\usepackage[T1]{fontenc}
\usepackage{ae,aecompl}
\usepackage[dvips]{graphicx}
\usepackage{subfigure}
\usepackage{xcolor}
\usepackage{array}
\usepackage{rotating}
\usepackage{colortbl}
\usepackage{tikz}
\usetikzlibrary{patterns}
\usepackage{float}
\usepackage{multirow}
\usepackage{dsfont}

\usepackage{amsthm}
\usepackage{amsmath}
\usepackage{amssymb}
\usepackage[utf8]{inputenc}
\usepackage{changepage}
\usepackage[toc,page]{appendix}
\usepackage{verbatim}

\definecolor{webgreen}{rgb}{0,0.4,0}
\definecolor{webbrown}{rgb}{0.6,0,0}
\definecolor{purple}{rgb}{0.5,0,0.25}
\definecolor{darkblue}{rgb}{0,0,0.7}
\definecolor{darkred}{rgb}{0.7,0,0}
\definecolor{darkgreen}{rgb}{0,0.7,0}
\usepackage[pdfborder=false]{hyperref}
\hypersetup{colorlinks,citecolor=darkred,filecolor=black,linkcolor=darkblue,urlcolor=webgreen,pdfpagemode=none,
pdfstartview=FitH}

\usepackage{sectsty}
\sectionfont{\centering\normalfont\scshape}
\subsectionfont{\centering\normalfont}

\begin{document}
\begin{spacing}{1.5}

\title{Social learning via actions in bandit environments\footnote{I am thankful to Ian Ball, Stephen Morris, and Kartik Vira for detailed comments and feedback. I also thank Anne Carlstein and Nathaniel Hickok for insightful suggestions.}}

\author{Aroon Narayanan}

\date{}

\maketitle

\begin{abstract}
    I study a game of strategic exploration with private payoffs and public actions in a Bayesian bandit setting. In particular, I look at \textit{cascade} equilibria, in which agents switch over time from the risky action to the riskless action only when they become sufficiently pessimistic. I show that these equilibria exist under some conditions and establish their salient properties. Individual exploration in these equilibria can be more or less than the single-agent level depending on whether the agents start out with a common prior or not, but the most optimistic agent always underexplores. I also show that allowing the agents to write enforceable ex-ante contracts will lead to the most ex-ante optimistic agent to buy all payoff streams, providing an explanation to the buying out of smaller start-ups by more established firms. \\

\end{abstract}

\newpage

\newtheorem{definition}{Definition}[]
\newtheorem{question}{Question}[]
\newtheorem{observation}{Observation}[]
\newtheorem{conjecture}{Conjecture}[]
\newtheorem{claim}{Claim}[]
\newtheorem{lemma}{Lemma}[]
\newtheorem{proposition}{Proposition}[]
\newtheorem{corollary}{Corollary}[]
\newtheorem{theorem}{Theorem}[]

\newcounter{example}[section]
\newenvironment{example}[1][]{\refstepcounter{example}\par\medskip
   \noindent \textbf{Example~\theexample. #1} \rmfamily}{\medskip}

\section{Introduction}

Many choices that economic agents face, be it a consumer deciding whether to try a new restaurant or a firm deciding whether to invest in a new technology, involve deciding between a risky and a safe option. The standard way to model this dichotomy is in the form of a \textit{bandit} problem, in which an agent faces the choice of choosing between two arms to pull, representing the two actions. Moreover, many of the environments that economic agents operate in are social - they explore new options at the same time as other agents who are doing the same. This is the setting I study, with the assumption that agents can see each others actions but not each others payoffs. For example, when deciding whether to try a new restaurant, a consumer may observe whether or not the restaurant is busy i.e. whether people keep returning to it, but probably not whether they enjoyed their meal (ignoring for this exercise the small subset that posts reviews online).

In particular, my interest is in an intuitive type of equilibria, which I term \textit{cascade equilibria}. In these, agents will keep taking the risky action until they become sufficiently pessimistic about its value, and then stop playing it forever, unless they observe someone else taking the risky action longer than they were supposed to. My main results revolve around establishing conditions under which such equilibria exist, and characterizing their salient properties.

The existence results bring out the crucial dilemma that agents face when exploring in a social context. If an agent take the risky action today, she can expect to get some immediate payoff and also some information about its worth, but at the risk of bearing some bad payoff in case the risky arm is actually detrimental. However, she could also take the safe action today and wait for the actions of her co-explorers to give her some indication about the right action to take. Cascade equilibria exist under conditions where the former outweighs the latter, leading agents to choose to explore rather than free-ride.

I also establish comparative statics for these equilibria. I allow for heterogenous priors and show that the most optimistic agent is always the last to stop exploring, but it need not always be the case that more optimistic agents explore more than less optimistic agents. I then compare exploration in cascade equilibria to the single-agent setting. If agents start out with a common prior, each agent explores less than what she would have, had she been exploring alone. With a heterogenous prior, however, agents may actually explore more than their single-agent level, in order to prevent more optimistic agents from becoming pessismistic and hence starting a cascade too soon, because they value the information that they can get from the exploration of others.

Finally, I show that allowing agents to write enforceable ex-ante contracts will lead to the most optimistic agent buying all of the payoff streams and then implementing her own efficient level of exploration. This phenomenon can be interpreted as the more established, better informed firm buying out or investing equity in smaller start-ups in the same space.

\textbf{Background and related literature}: The bandit problem, in particular the single-agent multi-arm bandit, has a vast literature in computer science and economics. In a Bayesian setting where one is comfortable assigning a prior to the payoff from each arm, the solution was derived first by \cite{gittens}. The agent assigns an \textit{index} to each arm every period, and pulls the arm with the highest index. When faced with one risky action and one safe action, the outcome under optimal play has a cutoff flavour - the agents chooses to take the risky action until some predetermined cutoff.

In a multi-agent setting, a third component of strategic benefit comes into play along with the immediate payoff and the option value of the risky action. This was first studied by \cite{bolton1999strategic}, who showed that a free-riding effect dominates in Markov equilibrium, leading agents to explore less in the social setting than they would have done had they been exploring alone. In their continuous-time model, both payoffs and actions are public information, which overloads the information benefit of the social setting - an agent can effectively depend on the exploration of others for her own benefit, introducing a free-rider problem. Moreover, the unique Markov equilibrium features complicated mixed strategies, rather than the simple cutoff strategies that were optimal in the single agent setting. However, \cite{horner2021} shows that this is an artefact of the solution concept rather than the problem itself, by showing that efficiency is in fact attainable under some conditions in strongly symmetric equilibria and the best equilibrium payoff is always achievable using cutoff strategies.

\cite{rosenberg2007social} study a discrete time setting with two players and private payoffs, imposing the restriction that playing the safe arm is irreversible, and find that all equilibria are cutoff equilibria. \cite{heidhues2015strategic} show that although Nash implementing the efficient outcome with common priors and observed payoffs is impossible even when the agents can communicate with teach other, it can be implemented under some conditions in sequential equilibrium when payoffs are unobserved and players can communicate with each other. They are silent on the case without communication.

Another strand of literature that my work closely relates to is the herding literature, started by \cite{banerjee1992simple} and \cite{bikhchandani1992theory}. The idea of information cascades is a running theme in this literature, with myopic agents ignoring their information when the information from predecessor actions are too valuable. The closest paper from this literature is \cite{mossel2015strategic}, who introduce forward-looking agents in a network setting and show that full learning can happen in the limit in \textit{egalitarian} networks.

\section{Model}

Each agent in a group of $n$ agents faces a choice between playing one of two actions, labelled $r$ (risky) and $s$ (safe), at time periods $t = 1,2...$. The state of the world, $\theta \in \{0,1\}$, is drawn before start of play and is common for all agents. The safe action gives a constant payoff of 0 always. The risky action's payoff each period is i.i.d. across players conditional on $\theta$ - it pays $X_h > 0$ with probability $\pi$ and pays $X_l < 0$ otherwise if $\theta = 1$, and pays $X_l < 0$ always if $\theta = 0$. Assume $E_1 = \mathbb{E}[X|\theta = 1] > 0$. Note that $E_2 = \mathbb{E}[X|\theta = 0] = X_l < 0 $. I informally say that an agent ``explores" when she plays the risky action.

Agent $i$ has prior $q_i$ about the state at time $t=1$. I denote $q_i(1)$ by $q_i$ for ease of notation i.e. $q_i$ is the probability that agent $i$ places on $\theta = 1$. I assume that the priors of the agents are commonly known. Wlog, I label the agents $1,...,n$ in decreasing order of $q_i$. The belief of agent $i$ about the state at time $t$ after history $h_t$ is written as $p_i(t, h_t)$. Agents observe the actions of other agents, but payoffs are private.

In this simple model, it is straightforward to calculate the cutoff belief at which a single agent exploring alone would stop exploring, by simply equating the benefit from exploring just one more period to the benefit from stopping forever:

\begin{equation*}
    p^a = \frac{(1-\delta)E_0}{(1-\delta)(E_1 + E_0) + \delta \pi E_1}
\end{equation*}

Starting at any prior $p$, this gives a time $\tau^a(p)$ until which the agent would explore. Similarly, given that the agent has $n$ risky actions to take, she stops at a cutoff:

\begin{equation*}
    p^e = \frac{(1-\delta)E_0}{(1-\delta)(E_1 + E_0) + \delta n \pi E_1}
\end{equation*}

This again corresponds to a time $\tau^e(p)$ until which the agent would explore. Note that $\tau^e(p)$ could be more or less than $\tau^a(p)$, since in the $n$ action case the agent gets $n$ draws per period instead of just $1$. But since $p^e < p^a$, the total number of draws that the agent takes in the $n$ action case will be more than that she takes when can take one action, so $n\tau^e(p) \geq \tau^a(p)$.

The focus of this paper will be on an intuitive class of equilibria, which I call \textit{cascade equilibria}. In such equilibria, on the equilibrium path, agents start out exploring and if an agent draws a success at any point, she plays the risky action forever. If however she draws only failures, then she will continue exploring till some pre-determined time, after which she will stop, unless she sees someone else explore past the pre-determined time that they were supposed to stop.

\begin{definition}
I call a weak perfect Bayesian equilibrium, $s$, a cascade equilibrium if $s$ induces the following ``cascade" outcome on the equilibrium path: there are cutoffs $(\tau_i)_{i \in N} \geq 1$ such that  
\begin{enumerate}
    \item agent $i$ plays the risky action until $t = \tau_i$
    \item conditional on all agents $j$ with cutoffs $\tau_j < \tau_i$ having already switched to playing the safe action, agent $i$ switches to playing the safe action at $t = \tau_i$ if and only if she has received $X_l$ every period that she played the risky action
    \item if some agent $j$ does not switch at her designated cutoff $\tau_j$, all agents play the risky action forever from $t = \tau_i + 1$
\end{enumerate}
\end{definition}

\begin{figure}
    \centering
    \includegraphics[width=.6\textwidth]{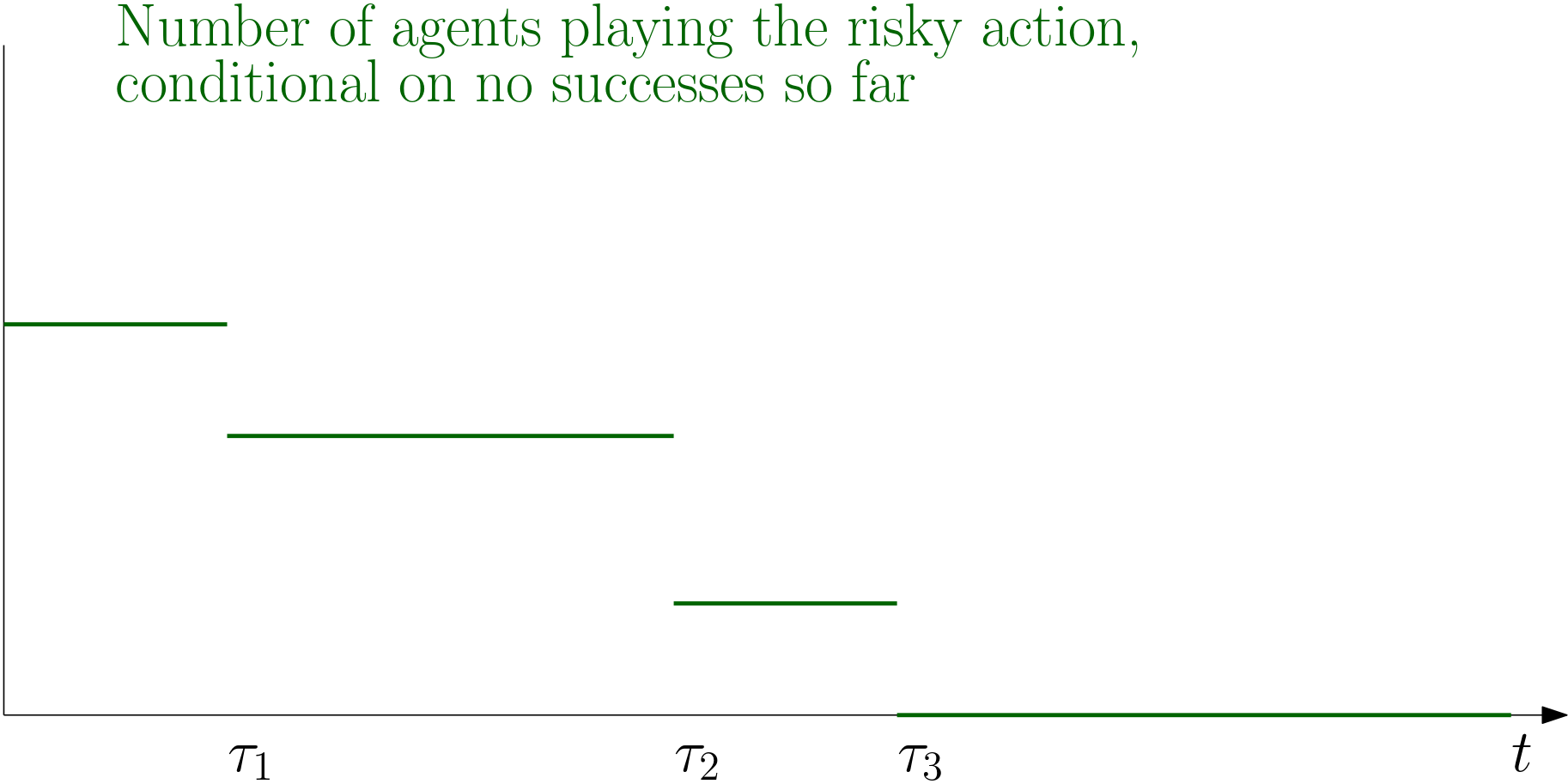}
    \caption{Cascade equilibrium outcomes}
    \label{fig:ind_hyp_straight}
\end{figure}

Note that cascade equilibria are described by the outcomes that are observed on path, remaining agnostic towards the equilibrium strategies that produce these outcomes. The primary motivation for defining them this way is that our interest lies in these outcomes themselves rather than the strategies employed to reach these outcomes.

\section{Equilibrium exploration}

\subsection{Existence of cascade equilibria}

First, I show that cascade equilibria are guaranteed to exist under some conditions. This is important because otherwise any results I prove using these equilibrium dynamics that I conjecture could be vacuously true. Since these are just sufficient conditions, I do not rule out that they may not exist outside of these conditions, but a more general proof will be a focus of future continuation of this work.

\begin{proposition}
With a common prior, that is when $p_i = p_j$ for all $i,j$, the following is a sufficient condition for the existence of cascade equilibria:

\begin{equation}\label{existenceineq}
    \frac{E_1}{E_0} \geq \frac{1-\pi}{\delta + \pi - (1-\pi)(1 + \frac{\delta \pi}{1-\delta})}
\end{equation}
\end{proposition}

\begin{proof}

Consider the following strategy profile - for any history $h_t$, each agent plays the risky action till her belief drops below $p_a$, after which she switches to the safe action and plays the safe action as long as her belief stays below $p_a$. If at any history her belief is $1$, she plays the risky action forever. Drawing successes and failures lead to standard Bayesian updating. If at any history an agent switches to the safe action from the risky action, the other agents infer that the deviating agent drew only failures when playing the risky action and update their beliefs about the good state according to Bayes rule. If any agent plays the risky action at any period where her beliefs can only either be $1$ or below $p_a$, other agents infer that it must be $1$ and everyone plays the risky action forever after. It is straightforward to check that these beliefs are consistent on path. As a sidenote, note that this is a symmetric strategy profile, and on path everyone has the same cutoff $\tau$.

Thus we are left to check sequential rationality of the strategies using the one shot deviation principle. Note that given a profile of beliefs, the equilibrium strategies require the same path of play, regardless of the history that led to the profile of beliefs. So it is sufficient to check one shot deviations at all possible profile of beliefs. Consider an agent's incentives at some belief $p_i > p_a$ to play the risky action as prescribed by the strategy vs deviating for one period and playing the safe action. If $p_i$ is such that the agent's beliefs will fall below $p_a$ after observing one more failure, then whether she deviates or not, she'll play the safe action from the next period if she draws no success, so she faces the single agent incentives, and hence she will find it optimal to play the risky action. For $p_i > p_a$ such that this is not the case, we have $p_i \geq \frac{p_a}{p_a + (1-\pi)(1-p_a)}$ by Bayes' rule, and the difference in payoffs from risky vs safe is bounded below by $(1-\delta)p_i[E_1 - \frac{(1-p_i)}{p_i}E_0] - p_i \delta^2 E_1$. The first term is the minimum additional payoff from playing the risky action this period. The second term in the maximum additional payoff from playing the safe action this period, obtained when playing the safe action makes the other agents reveal their draws via their actions, leading to the agent being perfectly match the state two periods from now. This bound is non-negative due to the inequality in the claim. When the agent's belief $p_i < p_a$, if she plays the risky action, everyone else plays the risky action forever after so she gains no information externality, and hence she will find it better to play the safe action.
\end{proof}

The interpretation of this sufficient condition in words relates to the trade off that the agent faces in terms of the informational benefit of other agents' exploration versus the informational benefit of taking the risky action. The left hand side quantifies the ratio of payoff benefits from the arms, and the right hand side connects it to a threshold based on informational considerations. The right hand side of the inequality, signifying the threshold that the payoff benefits must surpass, is low when $\pi$ is high and $\delta$ is low. In effect, a higher value for $\pi$ indicates that the risky action provides a signal of its worth more often, while a lower value of $\delta$ means that the agents place lower value of the informational benefit of other agent's exploration (since it always comes in the future). 

A concern with showing existence only for common prior profiles would be that such profiles are measure zero in the space of belief profiles. I show below that the common prior assumption can be relaxed for this result, for priors that are sufficiently close to some common prior. This is because generically, the agents will not be indifferent between pulling the safe and risky arms when they pull the risky arm for the last time, and will also not be indifferent between the two arms when they pull the safe arm for the first time. This slackness allows for a small variation in priors while keeping incentives the same.

\begin{proposition}
For any generic common prior belief profile $q$ and any set of parameters such that Equation \ref{existenceineq} is satisfied, there exists an open ball $B(q)$ around $q$ such that for any $p \in B(q)$, there is a cascade equilibrium with prior $p$.
\end{proposition}

\begin{proof}

Denote by $V_i(t,p_i)$ the continuation value of agent $i$ at time $t$ with belief $p_i$. Denote by $p(t,q_i)$ the belief of an agent at time $t$ given that she started out with a prior $q_i$ at time $t=0$, and then drew only failures till time $t$. Note that the equilibrium in Proposition 1 was symmetric, so on path agents will explore till some time $\tau$ and then all together stop exploring if none of them drew a success. 

The continuation value in the last period that they explore can be written as 

\begin{equation*}
    V(\tau,p(t,q_i)) = p(t,q_i) E_1 + (1-p(t,q_i)) E_0
\end{equation*}

Since $p(t,q_i)$ is derived by Bayes rule, it is continuous in $q_i$. And hence so is $V(\tau,p(t,q_i))$. Now suppose the common prior profile $q$ is such that agents are not indifferent between arms when stopping exploration i.e. there exists $\rho$ such that the value from playing the risky arm for the last time $V_i(\tau, p(\tau,q_i)) > \rho$ and the value from playing the risky arm the next time period $V_i(\tau+1, p(\tau+1,q_i)) < -\rho$. Then by continuity there is some $\epsilon > 0$ such that for priors $p_i \in (q_i - \epsilon, q_i + \epsilon)$, we have that $V_i(\tau, p(\tau,p_i)) \geq 0$ and $V_i(\tau+1, p(\tau+1,p_i)) \leq 0$. Thus agent $i$ will continue to find it optimal to explore till $\tau$ and then stop. 
\end{proof}

The proposition above establishes that the existence of cascade equilibria is not restricted to a measure zero subset of the space of prior profiles, since they exist in open balls around the set of common prior beliefs. Hence it is not something that one can rule out generically. However, it does leave open the question of whether one can expect them to exist for a significant proportion of prior beliefs.

\subsection{Properties of cascade equilibria}

Cascade equilibria were motivated by the intuitive way that agents explore in them, so the next logical step is to say more about the properties of agents' exploration. Since their exploration is quantified by the cutoff $\tau_i$, this is the relevant quantity that we must analyze.

A natural conjecture that one might have is that $\tau_i$ are ordered by agent optimism i.e. agents that start out with a higher $q_i$ will have a larger $\tau_i$. First I establish that this intuition holds for the most optimistic agent (labelled agent $1$ wlog).

\begin{proposition}\label{mostoptlast}
In any cascade equilibrium, $\tau_1 \geq \tau_j$.
\end{proposition}

\begin{proof}
Suppose  $\tau_1 < \tau_j$ for some agent $j \neq 1$ i.e. agent $1$ stops exploring before agent $j$. Wlog let agent $j$ be the last to switch to the safe action. Then, conditional on both of them getting only failures, we have that at $t = \tau_j$, $p_j(t) < p_1(t)$ since agent $j$ started out with a more pessimistic prior and has explored more. Moreover both agent $1$ and agent $j$ face single agent incentives at $t = \tau_j$ since neither agent's action influences how the other will act in the future. But then, if it is best for agent $j$ to play the risky action, it must be best for agent $1$ as well.
\end{proof}

This need not be true in general i.e. it is possible that $\tau_i$ < $\tau_j$ even though $p_i > p_j$. To see this, consider the following example.

\begin{example}\label{cutoffex}
Suppose there are three agents, and the priors are such that the first agent reaches the single-agent belief cutoff after observing $4$ failures, while the second and third reach it after $2$ draws, with the third agent a bit more pessimistic than the second. Specify the following strategies:

\begin{itemize}
    \item On the equilibrium path with no successes, agent $1$ plays the risky action for $3$ time periods, agent $2$ for $1$ period, and agent $3$ for $2$ periods.
    \item Off the equilibrium path, agents play the risky action until they reach their single-agent cutoff belief. If they see an agent play the risky action when they should've been playing the safe action, they infer that the state is good with probability $1$. If they see an agent switching from the risky to safe action, they infer that the agent only got failures so far.
\end{itemize}

I show in Appendix \ref{revorder} that this is an equilibrium for some parameter values. Note that agent $3$ explores more than agent $2$ even though her prior is more pessimistic. Both of them find it optimal to continue with this arrangement, for their own reasons. Agent $2$ knows that if she plays the risky action when she is supposed to switch, the others will assume that the state is good and always plays risky forever, which will lead to her losing access to their information. Agent $3$ knows that if she switches to the safe action, agent $1$ will play the risky action only until $t=2$, which is bad for her since it reduces her informational externality.
\end{example}

A corollary of Example \ref{cutoffex} is that cascade equilibria need not be unique for a given prior. Note that we could have switched the roles of agents $2$ and $3$, and continued to have the same incentives for staying on the equilibrium path as in the example. Agent $3$ would be willing to stop exploring early, since she is more pessimistic than agent $2$, and given that agent $3$ stops early, agent $2$ is willing to play the risky action for $2$ periods, since she is more optimistic than agent $3$.

This analysis carries over for comparative statics across different priors. We can compare $\tau_1$ across different prior profiles, but it is not possible to compare $\tau_i$ in general. This is established in the proposition and example that follow.

\begin{proposition}
Suppose for some $p,p'$, $p'_1 \geq p_1$. Then $\tau_1(p') \geq \tau_1(p)$.
\end{proposition}

\begin{proof}
By Proposition \ref{mostoptlast}, it must be the case that agent $1$ must be the last to stop exploring for both priors (note that agent $1$ refers to the most optimistic agent here so it is possible that this may not denote the same agent). Since the last agent to stop exploring faces single-person incentives, agent $1$ stops exploring at $\tau_a(p_1)$ with prior profile $p$ and at $\tau_a(p'_1)$ with prior profile $p'$. Given that $p_1 \geq p'_1$, we must have that $\tau_a(p_1) \geq \tau_a(p'_1)$. Thus the result follows.
\end{proof}

\begin{example}
It is not possible to compare the cutoffs of the other agents, even if we fix the priors of all but one of these agents. To see this, consider the setting of Example \ref{cutoffex}. In that example, agent $1$ explores $1$ period in equilibrium. However, as noted earlier, there is an equilibrium for the same prior in which the agent $1$ explores for $2$ periods. In fact, with a slightly smaller prior than in that example, agent $1$ would still explore for $2$ periods in an equilibrium, which means that she explores more for a smaller prior.

The intuition behind this is that the amount of exploration done by an agent depends on both her own beliefs and what other agents do in equilibrium. If the payoff that she anticipates from the actions of the other agents are sufficiently high, she will be willing to stop exploration at an earlier stage, even at a higher belief, compared to the case where the other agents aren't doing much socially beneficial exploration.
\end{example}

\subsection{Comparison to single-agent setting}

How does experimentation in the social setting compare to experimentation in the single agent setting? The answer to this question is nuanced. On the one hand, we can show that with a common prior, conditional on not having observed a success, an agent explores less in the social setting than if she were experimenting alone. To see this, define $\tau^a(p_i)$ as the time at which the agent switches to the risky action when experimenting alone, conditional on no successes. Given the set of cascade equilibrium cutoffs $T(p_i, p_{-i})$ for an agent in the social setting, also define

\begin{equation*}
    \Bar{\tau}(p_i, p_{-i}) = \max \{\tau \mid  \tau \in T(p_i, p_{-i})\}
\end{equation*}

\begin{proposition} \label{commexp}
If $p_i = p$ for all $i$, then $\tau^a(p) \geq \Bar{\tau}(p)$.
\end{proposition}

\begin{proof}
Suppose not. Then, there is some equilibrium of the social setting where an agent $i$ explores for $\Bar{\tau}(p) > \tau^a(p)$. Consider the time period $t=\Bar{\tau}(p)$, which is the last period that agent $i$ chooses the risky action, conditional on no successes so far. By the definition of $\Bar{\tau}(p)$, either all other agents have already switched to the safe action, or this will be the last period that they play the risky action, conditional on no successes. So the play of other agents cannot condition on what agent $i$ plays at $t=\Bar{\tau}(p)$, so her continuation payoff for both actions has the same information spillover benefits. Thus, she faces the single-agent trade-offs, which means she cannot play the risky action at $t=\Bar{\tau}(p)$.
\end{proof}

With heterogenous beliefs, however, this is no longer necessary, as I show in Example \ref{hetexp}. There can be equilibria in which agents explore more than when they explore alone. The intuition behind this is related to how agents anticipate the effects of their actions on the exploration of other agents. When an agent switches from the risky to the safe action, it immediately reveals the results of all her past draws to all other agents. This can cause the other agents to become pessimistic and stop exploring, even if they would've continued exploring without the release of such information. Thus agents can strategically explore a bit more than they would've in the single-agent setting, in order to benefit from the information spillover from the others. Note that for this story to be feasible, the other agents should be more optimistic, which necessitates heterogenous beliefs.

\begin{example} \label{hetexp}
Let $n=3$. The initial belief $p_{1}^{0}$ for the first agent is such that three bad draws drop her belief below $p^a$, while for the other two, it is such that five bad draws drop their beliefs below the single-agent cutoff. Specify the following strategies:

\begin{itemize}
    \item If $t \leq 4$ and everyone has played the risky action so far, everyone plays risky. 
    \item At $t=5$, agent $1$ switches to playing the safe action while agents $2$ and $3$ play the risky action.
    \item On and off path, agents $2$ and $3$ switch to the safe arm when their belief drops below $p^a$. Off path, agent $1$ switches to playing the safe action when her belief drops below $p^a$.
\end{itemize}

In Appendix \ref{hetexpapp}, I show that this strategy profile is indeed a PBE for some values of the parameters. Note that in this equilibrium, agent $1$ explores for $4$ periods, which is more than what she would have explored if she had been exploring alone.

However the logic from Proposition \ref{commexp} carries through for the most optimistic agent, labelled agent $1$ wlog, if such an agent exists uniquely. This is because in any cascade equilibrium, she must be the last agent to explore. But then, she cannot be exploring more than her single agent levels. 

\begin{proposition}\label{hetexp}
In any cascade equilibrium, $\tau^a(p_1) \geq \tau_1$.
\end{proposition}

\begin{proof}
The proof for this carries through in the same way as the proof for Proposition \ref{commexp}, since by Proposition \ref{mostoptlast}, agent $1$ is the last to stop exploring. 
\end{proof}
\end{example}

\subsection{Exploration with contracts}

Suppose agents now have access to money, using which they can write enforceable formal contracts ex ante\footnote{I assume away betting on the underlying state, which can lead to agents offering and accepting unbounded bets.}. I assume that there is a contracting stage before exploration, in which agents have access to as much financing as they
expect to get in the best case on the path of play. This can be interpreted as there being a financial institution that can learn the prior of the agent by verifying the private information that an agent has regarding the value of the project, and then providing the agent with financing up to the expected value under the agent's prior.

\begin{proposition}
Suppose there exists a most optimistic agent $i$ i.e.  $p_i > p_j$ for all $j \neq i$. Then if a contract is written ex ante, the outcome with money will be that agent $i$ owns all payoff streams and explores for $\tau^e(p_i)$ periods.
\end{proposition}

\begin{proof}
Consider an agent with prior $p$ about the state. For now, let each of the $n$ payoff streams be represented by a bandit. Her maximum payoff is when she can utilize all bandits to learn about the state and maximize the total payoff from all bandits - this is the efficient level of exploration. The efficient cutoff belief at which the agent stops exploring, i.e. when the agent is indifferent between pulling just one arm and stopping all exploration, is given by 

\begin{equation*}
    p^e = \frac{(1-\delta)E_0}{(1-\delta)(E_1 + E_0) + \delta n \pi E_1}
\end{equation*}

Given the prior $p$, this corresponds to a cutoff $\tau^e(p)$ at which the agent stops all exploration. Note that $\tau^e(p)$ is increasing in $p$ as a step function, i.e. it is constant in intervals and then jumps up by $1$ at the right end point of the intervals.

Thus, the total payoff that the agent ex ante expects to get is given by:

\begin{align*}
    \sum_{t=1}^{\tau^e(p)} & n \{ p (1-\delta) E_1 - (1-p)  (1-\delta) E_0 \} + np[1-(1-\pi)^{{n\tau^e(p)}}]\delta^{\tau^e(p)+1}E_1 \\
    = & n(1 - \delta^{\tau^e(p)+1}) [ p E_1 - (1-p) E_0 ] + np[1-(1-\pi)^{{n\tau^e(p)}}]\delta^{\tau^e(p)+1}E_1   \\
\end{align*}

As $p$ increases, whenever $\tau^e(p)$ is constant, this is clearly increasing in $p$. At the cutoff where $\tau^e(p)$ jumps up, the agent with that cutoff belief is indifferent between exploring at the two levels, so she gets a higher payoff than those with a slightly smaller prior, since the payoff is increasing whenever $\tau^e(p)$ is constant. Thus, it is clear that the total payoff is increasing in $p$.

Since an agent with a higher prior can expect to earn a higher payoff and hence utilize a higher ex ante credit line, she can offer higher side payments to the other agents than any other agent can offer. The result then follows.

\end{proof}

This result predicts that the contracts that will be seen in practice will be ones where the most optimistic agent acquires the payoff streams of the less optimistic ones. Consider a setting where firms are developing new technology socially. Their priors at the beginning of the exploration stage are informed by their research ex ante, which leads to different signals of how lucrative the technology happens to be. In cases where the most established firm can research the technology the best ex ante and hence have the most optimistic signal, the result above amounts to saying that the most established firm will ``buy out" smaller startups. Hence it provides a belief-centric explanation as to why these buy outs are seen in practice at all. In many standard settings, results like the \cite{milgrom1982information} no trade theorem and \cite{myerson1983efficient} impossibility theorem establish the impossibility of trade under the common prior even when agents have private information, precluding such buy outs from being rationalized.

\section{Discussion}

\textbf{Existence of non-cascade equilibria}: In this work, I focus on equilibria in which behaviour is similar in spirit to the cutoff strategies that are prevalent in earlier work. There will, however, be other interesting outcomes, produced by strategies with complicated behaviour. For example, agents can switch from the risky to the riskless action in continuous coordination, releasing information about their own draws early on in return for information for other agents' draws, so that they can match their action to the state sooner. These run the risk of making the agents pessimistic earlier than necessary, reducing equilibrium exploration and hence reducing overall welfare. Yet another class of equilibria can have mixing on the equilibrium path. This class of equilibria will have the property that agents may switch from playing the risky action to playing the safe action and then again switch back, without having gained any additional information about the state or the draws of others. Thus although interesting from a theoretical viewpoint, they exhibit dynamics that unlikely to be observed in practice.

\textbf{Welfare properties of cascade equilibria}: Cascade equilibria were informally motivated by the observation that they feature delayed information release, but a result that formalizes the positive implications of this would need to establish that cascade equilibria indeed have superior welfare or exploration properties. An objective of future work could be to formally establish that cascade equilibria indeed have superior welfare and exploration properties, similar in spirit to the results in \cite{horner2021}.

\newpage

\bibliographystyle{apalike}
\bibliography{main}

\newpage

\appendix

\section{Examples}

\subsection{Reversed order of cutoffs}\label{revorder}

The only time period to check for deviations is time $t=2$. For agent $2$, the differential benefit from playing the risky versus the safe action is given by:
    
\begin{equation*}
    p_2 \pi E_1  - (1-p_2) (1-\delta) E_0 - I_2
\end{equation*}

where $I_2 > 0$ is the information benefit that the agent gets by staying on the equilibrium path.

For agent $3$, the differential benefit from playing the risky versus the safe action is given by:

\begin{equation*}
    p_3 \pi E_1  - (1-p_3) (1-\delta) E_0 + I_3^p - I_3^d
\end{equation*}

where $I_3^p = \delta^3 p_3 [1-(1-\pi)^3]E_1$ is the information benefit that the agent gets by staying on the equilibrium path and $I_3^d = \delta^2 p_3 [1-(1-\pi)^2]E_1$ is the information benefit from deviating.

Note that given the other parameters, $\pi$ can be chosen such that $I_3^p - I_3^d > 0$. Then, there are priors close enough for the two agents such that 

\begin{equation*}
    p_3 \pi E_1  - (1-p_3) (1-\delta) E_0 + I_3^p - I_3^d > 0 > p_2 \pi E_1  - (1-p_2) (1-\delta) E_0 - I_2
\end{equation*}

Thus agent $2$ will play the safe action and agent $3$ will play the risky action.

\subsection{Exploring more with heterogenous priors}\label{hetexpapp}

We check for deviations at the key decision nodes for the players:

\begin{enumerate}
    \item Agent $1$ at $t=4$: the agent's on path payoff is given by:
    
    \begin{equation*}
        p_1 [(1-\delta)E_1 + \delta \pi E_1 + \delta (1-\pi) (1 - (1-\pi)^{11}) \delta^2 E_1 ] + (1-p_1)[-(1-\delta)E_0]
    \end{equation*}
    
    while if she deviates, she gets
    
    \begin{equation*}
        p_1 [(1 - (1-\pi)^8) \delta^2 E_1 ]
    \end{equation*}
    
    The difference between these two is 
    
    \begin{equation*}
        p_1 [(1-\delta)E_1 + \delta \pi E_1 ] + (1-p_1)[-(1-\delta)E_0] - p_1 \delta^2 E_1  [(1-\pi) (1 - (1-\pi)^{11}) \delta - (1-(1-\pi)^8)]
    \end{equation*}
    
    Note that $p_1 [(1-\delta)E_1 + \delta \pi E_1 ] + (1-p_1)[-(1-\delta)E_0 = 0$ is solved by the single-agent cutoff belief, so this term can be made arbitrarily small by choosing the initial belief of agent $1$ appropriately, independent of $\delta$ and $\pi$. We can also choose $\delta$ close enough to $1$ so that it is sufficient to check that
    
    \begin{equation*}
        (1-\pi) (1 - (1-\pi)^{11}) - (1-(1-\pi)^8) > 0
    \end{equation*}
    
    which is true when $\pi$ is small.
    
    \item Agents 2 and 3 at t=4: the on path payoff for agent 2 (and symmetrically, agent 3) at t=4 is:
    
    \begin{equation*}
        p_2 [(1-\delta^2)E_1 + \delta^3 (1-(1-\pi)^2) (1-\delta) E_1 + \delta^4 (1 - (1-\pi)^7) E_1 ] + (1-p_2)[-(1-\delta^2)E_0]
    \end{equation*}
    
    while if she deviates, she gets 
    
    \begin{equation*}
        p_2 [(1 - (1-\pi)^4) \delta^2 E_1 ]
    \end{equation*}
    
    For $\pi$ small, the deviation payoff is small but the on path payoff is still large since the two agents are two draws away from their single-agent cutoffs.
    
    \item Agents 2 and 3 at t=5: Both agents face their single-agent incentives at $t=5$ since they will get the information of the other agent at $t=6$ no matter what they do, so they will both choose to play the risky action. 
    
\end{enumerate}

\end{spacing}
\end{document}